\documentclass{llncs}

\usepackage[english]{babel}
\usepackage{times}
\usepackage{amsmath, amssymb}
\usepackage{algorithm}
\usepackage[noend]{algorithmic}
\algsetup{indent=1.3em}
\usepackage{tikz}
\usepackage{pgfplots}
\usepackage{xspace}
\usepackage{multirow}
\usepackage{setspace} 
\usepackage{subfigure}
\usepackage{bold-extra}
\usepackage[margin=1.08in]{geometry}

\newcommand{\N}{{\mathbb N}}

\newcommand{\T}{\mathbb{T}}

\newcommand{\J}{{\cal J}}

\newcommand{\G}{{\cal G}}

\newcommand{\temp}[1]{\hat{#1}}

\newcommand{\tclocks}{{\textsc{T-Clocks}}\xspace}

\newcommand{\commenttt}[1]{$\rhd$~{\it #1}}
\newcommand{\codetitle}[1]{\medskip\STATE \underline{#1:}\vspace{2pt}}

\title{Building Fastest Broadcast Trees in Periodically-Varying Graphs}
\author{
\institute{}
 Arnaud Casteigts\inst{1}
 \and Paola Flocchini\inst{1}
 \and Bernard Mans\inst{2}
 \and Nicola Santoro\inst{3}
 \institute{University of Ottawa, Ottawa, Canada,\\
 \email{\{casteig,flocchin\}@site.uottawa.ca}
 \and 
 Macquarie University, Sydney, Australia,\\
 \email{bernard.mans@mq.edu.au}
 \and
 Carleton University, Ottawa, Canada,\\
 \email{santoro@scs.carleton.ca}}
}
\begin{document}
\maketitle

\begin{abstract}
  Delay-tolerant networks (DTNs) are characterized by a possible absence of end-to-end communication routes at any instant. Still, connectivity can generally be established over time and space. The optimality of a temporal path (journey) in this context can be defined in several terms, including topological (e.g. {\em shortest} in hops) and temporal (e.g. {\em fastest, foremost}). 
The combinatorial problem of computing shortest, foremost, and fastest journeys {\em given full knowledge} of the network schedule was addressed a decade ago (Bui-Xuan {\it et al.}, 2003). A recent line of research has focused on the distributed version of this problem, where foremost, shortest or fastest {\em broadcast} are performed without knowing the schedule beforehand.
In this paper we show how to build {\em fastest} broadcast trees ({\it i.e.,} trees that minimize the global duration of the broadcast, however late the departure is) in Time-Varying Graphs where intermittent edges are available periodically (it is known that the problem is infeasible in the general case even if various parameters of the graph are know). We address the general case where contacts between nodes can have arbitrary durations and thus fastest routes may consist of a mixture of {\em continuous} and {\em discontinuous} segments (a more complex scenario than when contacts are {\em punctual} and thus routes are only discontinuous). Using the abstraction of \tclocks to compute the temporal distances, we solve the fastest broadcast problem by first  learning, at the emitter, what is its time of {\em minimum temporal eccentricity} ({\em i.e.} the fastest time to reach all the other nodes), and second by building a {\em foremost} broadcast tree relative to this particular emission date.
\end{abstract}

\section{Introduction}

Highly-dynamic networks, in particular Delay-Tolerant Networks, are characterized by a possible absence of end-to-end communication routes (paths or {\em direct journeys}) at any time. In most cases, however, communication can still be achieved over time and space using store-carry-forward-like mechanisms over disconnected routes ({\em indirect journeys}). This singularity led researchers to develop a range of routing techniques based for example on prior knowledge of the dynamics~\cite{BFJ03,JFP04}, probabilistic information~\cite{LDS03}, encounter-based strategies~\cite{DubGV03,GroV03,JLW07}, or contact statistics~\cite{KMR05}. On the other hand, considering the time-dimension led to the extension of many graph theoretic and analytical concepts including paths and reachability~\cite{Berman96,KKK00}, distance~\cite{BFJ03,JMR10}, diameter~\cite{ChMMD08}, tree width~\cite{MM11}, connectivity~\cite{AE84,BF03}, or necessary conditions~\cite{CCF09,KLO10}. Of particular interest in this paper are the concepts of {\em journeys} (temporal paths) and {\em temporal distance}. 

The fact that connectivity takes place over time implies that the optimality of a given journey does not depend anymore on the sole number of hops separating the nodes. Short journeys might have a long duration, and long ones be comparably fast. In fact, at least three optimality metrics could be considered for a journey, as pointed out in~\cite{BFJ03}: being {\em shortest} (least number of hops), {\em foremost} (earliest arrival date), or {\em fastest} (minimum time spent between departure and arrival, however late the departure is). The problem of computing shortest, foremost, and fastest journeys was addressed in~\cite{BFJ03} in the context of centralized algorithms using full prior knowledge of the network schedule. 

In a recent line of research~\cite{CFMS10,CFMS11}, the authors started to look at the {\em distributed} variant of this problem, namely performing foremost, shortest, and fastest {\em broadcast} in time-varying graphs {\it without} knowledge of the schedule. Specifically, we asked what assumptions could make each of these broadcasts possible, and what others would make the broadcast trees reusable for subsequent broadcasts. Within the set of assumptions considered, the easiest problem regarding {\em feasibility} turned out to be foremost broadcast, while the easiest regarding {\em reusability} is shortest broadcast~\cite{CFMS10}. Due to its intrinsic features, fastest broadcast remained unfeasible in all the contexts considered, which motivates us to consider stronger assumptions such as {\em periodicity}. The periodic assumption holds in practical scenarios such as transportations {\it e.g.,} where entities have periodic movements (satellites, trains, or buses). Previous works in periodic DTNs include exploration~\cite{FKMS11,FMS09,IW11} or scalable routing~\cite{KerO09,LW09b}.

We are tackling the problem in two steps. First, having the emitter decide {\em when} broadcast has the potential to be the {\em fastest}, then build a {\em foremost} broadcast tree relative to (any of) the corresponding emission date. We show that the latter is trivial, and thus mainly focus on the way the emitter could process the optimal date to initiate a broadcast -- {\em i.e.}, its moment of {\em minimum temporal eccentricity}. This processing is achieved by using a recently introduced primitive, \tclocks, which allows nodes in an arbitrary dynamic network to learn temporal lags a {\it posteriori} (that is, from a {\em receptor-based} perspective). Temporal lags relative to the emitter are thus first learnt by destination nodes ({\it i.e.,} all the other nodes), then aggregated back to the emitter in an opportunistic way so that the emitter learns the complete evolution of its temporal eccentricity over one period $p$. Since the network is periodic, this information suffices to determine the emitter's temporal eccentricity at any time in the future. One of the main assets of our solution is that it addresses the general (and more realistic) case where contacts between nodes can have arbitrarily durations and possibly overlap in time with other contacts. This feature renders temporal distances substantially more complex to process, because they involve the co-existence of continuous and discontinuous routes in the network (a vast majority of DTN-related results assume punctual contacts and thus can only deal with discontinuous routes).

The paper is organized as follows: in Section~\ref{sec:model}, we describe the model and terminology used in the paper. Then Section~\ref{sec:tclocks} present the building block upon which our algorithm relies, that is, the temporal-lags vector clocks (\tclocks). Finally, Section~\ref{sec:fastest} describes the algorithm, proves its correctness, and illustrate it with an example scenario.

%%% MODEL
\section{Background}
\label{sec:model}

\subsection{Model and assumptions}
%%%    TIME-VARYING GRAPH
Consider a set $V$ of $n$ {\em nodes}, making contacts with each other over a (possibly infinite) time interval ${\cal T} \subseteq \T$, called {\em lifetime} of the system; the temporal domain $\T$ corresponds here to $\mathbb{R}^+$ (continuous-time system). Let the contacts between nodes define a set of intermittently available undirected edges $E\subseteq V^2$ such that $(x,y)\in E \Leftrightarrow x$ and $y$ interact at least once in ${\cal T}$.

Following~\cite{CFQS11}, we represent the network as a {\em Time-Varying Graph} (TVG, for short) $\G=(V,E,{\cal T},\rho, \zeta)$, where $\rho : E \times {\cal T} \rightarrow \{0,1\}$, called {\em presence} function (or {\em schedule}), indicates whether a given edge is available at a given time, and $\zeta \in \T$, called {\em latency}, indicates the time it takes to propagate a message over an edge. In this work, we assume the latency is constant for all edges and presence times (though in general it could be defined as a function) and known to the nodes. For this paper, we consider Periodically-Varying Graphs (PVG, for short) where intermittent edges are available periodically. That is, the schedule $\rho$ is unknown to the nodes, but assumed periodic: $\forall e\in E, \forall t\in {\cal T},\forall k\in \N,  \rho(e,t)=\rho(e,t+kp)$, and the period $p$ is known to the nodes.
If a message is sent less than $\zeta$ time units before the disappearance of an edge, it is lost. Given an edge $e$, we allow the notation $\rho_{[t_1,t_2)}(e)=1$ to signify that $\forall t \in [t_1,t_2), \rho(e,t)=1$ ({\it i.e.,} the edge is continuously present during interval $[t_1,t_2)$).

%%%    JOURNEYS
Given a time-varying graph $\G=(V,E,{\cal T},\rho, \zeta)$, we denote by $G=(V,E)$ its {\em underlying graph}, that is, in a sense, the {\em static} counterpart of $\G$.
A sequence of couples $\J=((e_1,t_1),$ $(e_2,t_2) \dots,$ $(e_k,t_k))$,  where $e_1, e_2,...,e_k$ is a walk in $G$, $ t_i +\zeta  \leq t_{i+1}$ and $\rho_{[t_i,t_i+\zeta)}(e_i)=1$ for all $1\leq i <k$ is called a {\em journey} in $\G$. 
 We denote by $departure(\J)$, and $arrival(\J)$, the starting date $t_1$ and last date $t_k+\zeta$ of $\J$, respectively.
 Journeys can be thought of as {\em paths over time} from a source to a destination and thus have both a {\em topological} and a {\em temporal} lengths.
 The {\em topological length} of $\J$ is the number $|\J|_h= k$ of couples in $\J$ (i.e., number of {\em hops}), and its {\em temporal length} is its duration $|\J|_t  =  arrival(\J) - departure(\J) = t_k - t_1 +\zeta$. For example the journey $((ac,2), (cd,5))$ in Figure~\ref{fig:example} has a topological length of $2$, and a temporal length of $3+\zeta$ time units.

Let us denote by $\J^*_\G$ the set of all journeys in  time-varying graph $\G$, and by  $\J^*_{(u,v)} \subseteq \J^*_\G$ those journeys starting at node $u$ and ending at node $v$.

\begin{figure}[h]
  \centering
  \begin{tikzpicture}[scale=1.6]
    \tikzstyle{every node}=[draw, circle, minimum size=11pt, inner sep=0pt]
    \path (0,0) node (a){$a$};
    \path (a)+(0, 1) node (b){$b$};
    \path (a)+(1,.5) node (c){$c$};
    \path (c)+(1.2,0) node (d){$d$};
    \tikzstyle{every node}=[font=\scriptsize, inner sep=1pt]
    \draw (a)--node[midway, left]{$[1,3)$}(b);
    \draw (a)--node[midway, xshift=-4pt, yshift=-1pt, below right]{$[2,5)$}(c);
    \draw (b)--node[midway, xshift=-4pt, yshift=1pt, above right]{$[0,4)$}(c);
    \draw (c)--node[midway, above]{$[5,6)$}(d);
  \end{tikzpicture}
  \caption{\label{fig:example}A time-varying graph $\G$. ({\it Labels indicate when an edge is present}.)}
\end{figure}
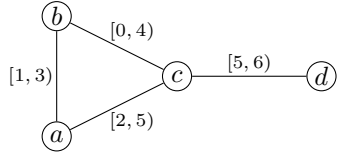

%%%    PROPERTIES OF JOURNEYS
We say that a journey is {\em direct} if the presence of every two successive edges overlap in time and their use follow on directly; it is said {\em indirect} otherwise. An example of direct journey in the graph of Figure~\ref{fig:example}  is $\J_1=\{(ab,2),(bc,2+\zeta)\}$. Examples of indirect ones
  include $\J_{2}=\{(ac,2),(cd,5)\}$, and  $\J_{3}=\{(ab,2),(bc,2+\zeta),(cd,5)\}$. 

%%%    ASSUMPTIONS
We assume nodes can {\em instantly} detect the appearance or disappearance of an incident edge, and associate dedicated operations in reaction to these events. Since the presence intervals are by convention right-open, we assume that the disappearance of an edge at time $t$ is always  handled {\em before} the appearance of another edge at date $t$, which is consistent with looking at journeys whose presence intervals strictly follow each other as indirect ones ({\it e.g.} $\J_{(a,d)}=\{(ac,5$$-$$\zeta), (cd,5)\}$ in the above example). 
Processing times are assumed negligible. 
The entities need not sharing a common global time, but their clocks must advance at a same rate. Finally, we assume nodes have unique identifiers.

\subsection{Definition of the problem}
\label{sec:distances}

%%%    TEMPORAL DISTANCE AND FOREMOST JOURNEY
As mentioned above, the length of a journey can be measured both in terms of {\em hops} or {\em time}. This gives rise to two distinct definitions of distance in a graph~$\G$:
\begin{list}{\labelitemi}{\leftmargin=.5em}
  \item The {\em topological distance} from a node $u$ to a node $v$ at time $t$, noted 
  $d_{u,t}(v)$, is defined as $Min\{|\J|_h:\J \in \J^*_{(u,v)} \wedge departure(\J) \ge t\}$. For a given date $t$, a journey whose departure is $t'\ge t$ and topological length is equal to $d_{u,t}(v)$ is called {\em shortest} ;
  
  \item The {\em temporal distance} from $u$ to $v$ at time $t$, noted $\hat{d}_{u,t}(v)$ is defined as $Min\{arrival(\J):\J \in \J^*_{(u,v)} \wedge departure(\J)\ge t\}-t$. Given a date $t$, a journey whose departure is $t'\ge t$ and arrival is $t+{\hat d}_{u,t}(v)$ is called {\em foremost}, and one whose departure is $t' \ge t$ and temporal length is $Min\{{\hat d}_{u,t''}(v) : t'' \ge t\}$ is called {\em fastest}.
\end{list}

Informally, a {\em fastest} journey is one that minimizes the time spent between departure and arrival (however late the departure is). This metric is particularly relevant in communication networks whose medium is shared exclusively (e.g. to minimize the holding time), or transportation networks (to minimize a trip duration, even if this implies leaving later). The problem of computing shortest, fastest, and foremost journeys in delay-tolerant networks was solved in~\cite{BFJ03} as a centralized ({\it i.e.,} combinatorics) problem, given complete knowledge of $\G$. We consider a {\em distributed} variant of the problem, namely that of performing fastest broadcast, which we define as follows.

\paragraph{Fastest Broadcast:} Given an emitter $src$ and a message $m$, $m$ should be sent from $src$ to all the nodes in the network in such as way that the global duration between the first message emission (at $src$) and the last message reception (anywhere) is {\em minimized}. We also require that the emitter detects the termination, whether implicitly or explicitly, however this detection does not need to be the fastest.

\subsection{Working with temporal views}

%%%    VIEWS
A central concept in this paper is that of {\em temporal view}, introduced in~\cite{KosKW08} in a context of social network analysis. (This concept  was simply called ``view'' in \cite{KosKW08}, but we added the ``temporal'' adjective to avoid confusion since ``view'' has a different meaning in distributed computing, e.g. \cite{YamK96}.) The temporal view a node $v$ has of a node $u$ at time $t$, denoted $\phi_{v,t}(u)$, is the latest (i.e., largest) $t'\le t$ at which a message received by time $t$ at $v$ could have been emitted at $u$; that is, in our formalism,

\begin{center}
  \small
  $\phi_{v,t}(u)=$
  Max$\{departure(\J) : \J \in\J^*_{(u,v)}$ $ \wedge~ arrival(\J)\le t\}$.
\end{center}

There is a clear connexion between temporal distances and temporal views. Both actually refer to the same quantity seen from different perspectives: the temporal distance is a {\em duration} defined locally to an {\em emitter} at an emission date, while the temporal view is a {\em date} defined locally to a {\em receptor} at a reception date. In fact, we have
\begin{equation}
  \temp{d}_{u,t_e}(v) = t_r - \phi_{v,t_r}(u)
\end{equation}
where $t_e$ is an emission date, and $t_r$ is the corresponding earliest reception date.

This definition can be seen as a generalization of the one in~\cite{KosKW08} in which only punctual contacts were considered. As pointed out in~\cite{CFMS10}, the temporal view a node has of another is deeply impacted by the co-existence of direct and indirect journeys. Indeed, assuming arbitrary long contacts between nodes makes it possible for adjacent contacts to overlap in time and thus produce more complex patterns of time lag between nodes. 
Consider the plots in Figure~\ref{fig:overlapping}, showing an example of evolution of a temporal distance (from $a$ to $c$) and the corresponding temporal view (that $c$ has of $a$) in a very simple TVG.
Contrary to the case with punctual contacts --~where evolution occurs only in discrete steps~-- there is here a mixture of discrete and continuous evolution. (The reader is encouraged to spend a few minutes on this example as these concepts are essential in what follows.)

\begin{figure}[h]
  \centering
  \subfigure[A simple TVG]{
    \label{fig:basicgraph}
    \begin{tikzpicture}[scale=1.2]
      \tikzstyle{every node}=[draw,circle, minimum size=11pt, inner sep=0pt]
      \path (0,-.5) coordinate;
      \path (0,0) node[very thick] (a){a};
      \path (a)+(1,0) node (b){b};
      \path (a)+(2.8,0) node (c){c};
      \tikzstyle{every node}=[below,font=\scriptsize,inner sep=1pt]
      \draw (a)--node[above, yshift=1pt]{$[0,4)$}(b);
      \draw (b)--node[above, yshift=1pt]{$[1,3)$$\cup$$[5,6)$}(c);
    \end{tikzpicture}
  }%\hspace{20pt}
  \subfigure[Temporal distance from $a$ to $c$]{
    \label{fig:viewplot-d}
    \begin{tikzpicture}[scale=.6]
      \draw[->] (0,0)--(0,4);
      \draw[->] (0,0)--(5,0);
      \draw (0,1.2)--(.8,0.4)--(2.6,0.4)--(2.6,2.4)--(3.8,1.4) -- (3.8,3) edge [dashed](3.8,4);
      \tikzstyle{every node}=[font=\scriptsize, inner sep=3pt]
      \path (0,0) node[below left, inner sep=2pt] {$0$};
      \path (0,0.4) node[left] {$2\zeta$};
      \path (0,1.2) node[left] {$1$$+$$\zeta$};
      %\path (0,1.4) node[left] {$1$$+$$2\zeta$};
      \path (0,2.4) node[left] {$2$$+$$2\zeta$};
      \path (0,3) node[left] {$3$};
      \tikzstyle{every node}=[font=\scriptsize, inner sep=4pt]
      \path (.8,0) node[below] {$1$$-$$\zeta$};
      \path (2.6,0) node[below] {$3$$-$$2\zeta$};
      \path (3.8,0) node[below] {$4$$-$$\zeta$};
      \tikzstyle{every node}=[font=\footnotesize, inner sep=4pt]
      \path (0,4) node[right] {$\temp{d}_{a,t}(c)$};
      \path (5,0) node[right] {$t$};
      \draw (.8,-.1)--(.8,.1);
      \draw (2.6,-.1)--(2.6,.1);
      \draw (3.8,-.1)--(3.8,.1);
      \draw (-.08,.4)--(.08,.4);
      \draw (-.08,1.2)--(.08,1.2);
      %\draw (-.08,1.4)--(.08,1.4);
      \draw (-.08,2.4)--(.08,2.4);
      \draw (-.08,3)--(.08,3);
    \end{tikzpicture}
  }
  \subfigure[Temporal view that $c$ has of $a$]{
    \label{fig:viewplot}
    \begin{tikzpicture}[xscale=.4,yscale=.5]
      \draw[->] (0,0)--(0,5);
      \draw[->] (0,0)--(7,0);
      \draw (1,0)--(1,1)--(3,3)--(5,3)--(5,4)--(6,4)
      edge [dashed](7,4);
      \tikzstyle{every node}=[font=\scriptsize, inner sep=3pt]
      \path (0,0) node[below left, inner sep=2pt] {$0$};
      \path (0,1) node[left] {$1$$-$$\zeta$};
      \path (0,3) node[left] {$3$$-$$2\zeta$};
      \path (0,4) node[left] {$4$$-$$\zeta$};
      \path (1,0) node[below] {$1$+$\zeta$};
      \tikzstyle{every node}=[font=\scriptsize, inner sep=4pt]
      \path (3,0) node[below] {$3$};
      \path (5,0) node[below] {$5$};
      \tikzstyle{every node}=[font=\footnotesize, inner sep=4pt]
      \path (0,5) node[right] {$\phi_{c,t}(a)$};
      \path (7,0) node[right] {$t$};
      \draw (1,-.1)--(1,.1);
      \draw (3,-.1)--(3,.1);
      \draw (5,-.1)--(5,.1);
      \draw (-.08,1)--(.08,1);
      \draw (-.08,3)--(.08,3);
      \draw (-.08,4)--(.08,4);
    \end{tikzpicture}
  }
  \caption{\label{fig:overlapping}Temporal distance and temporal views as a function of time {\it (with $\zeta \ll 1$)}.}
\end{figure}
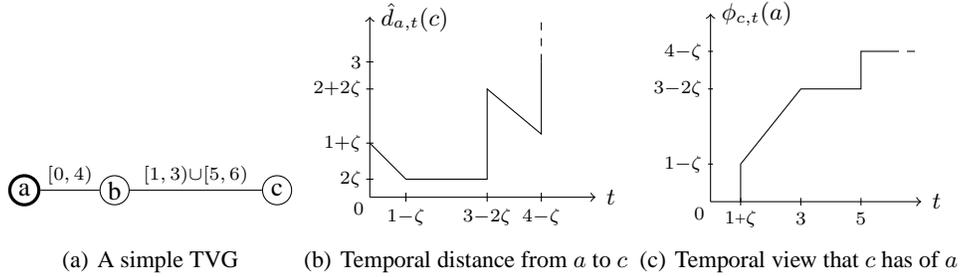

Direct journeys are often faster than indirect ones, but this is not necessarily the case (imagine a very long direct journey, versus a short indirect one whose edges traversals follow closely). As a result, the temporal view at anytime could be caused by either types of journey. Let us call {\em direct view} and {\em indirect view} the views resulting from the best direct  and indirect journeys at any given time (keeping in mind that the temporal view is always the max between both).

Direct and indirect views are of a different nature. Indeed, contrary to indirect journeys, which take place in a discrete way, direct journeys in general belong to a {\em continuum} of several such journeys, which induce {\em continuous} increases of the view. Looking again at the example on Figure~\ref{fig:viewplot} the (direct) view $c$ has of $a$ during $[1+\zeta,3)$ only depends on the {\em topological} length of the corresponding journey, here $2$ (node $c$ can receive any message emitted by $a$ between times $1-\zeta$ and $3-2\zeta$ after a lag of exactly $2\zeta$ time).
Since the duration of a direct journey depends only on its topological length, it is sufficient for a node $u$ to know exactly what is the length of the {\em shortest} direct journey currently arriving to it from another node $v$, called the {\em level} of $u$ with respect to $v$, to deduce the corresponding direct view in real time. This particular way of defining a level can be formalized as follows:
\begin{center}
  \small
  \begin{tabular}{c@{~}l}
    $level_{v,t}(u) =Min\{|\J|_h:\J\in \J^*_{(u,v)}$ $\wedge~ isDirect(\J) \wedge arrival(\J)=t\}$,
  \end{tabular}
\end{center}

\noindent where $isDirect(\J)$ is true iff $\J$ is a direct journey. By convention, $level_{v,t}(u)$ is considered to be $-\infty$ if no direct journey from $u$ is arriving to $v$ at time $t$. Let us insist on the fact that our notion of level is relative to the {\em reception} side, and is not concerned for example with the fact that some edges of the journey may have disappeared by reception time. What matters is to know whether messages could still be {\em currently} arriving from a given remote node through a given number of hops.

%\ \\
\section{Tracking temporal views with \tclocks}
\label{sec:tclocks}

This section is concerned with tracking in real-time the evolution of direct and indirect views. A powerful abstraction in this regard was provided in~\cite{CFMS11}. Precisely, an algorithm called \tclocks was provided that allows to continuously keep track of direct and indirect views in a delay-tolerant network, and allows to plug another algorithm on top of it to solve more concrete problems using this information. The example given in~\cite{CFMS11} was the construction of {\em foremost} broadcast trees in periodically-varying graphs. In this paper, we will use \tclocks to solve the construction of {\em fastest} broadcast trees in periodically-varying graphs, which uses a completely different approach compared to the foremost case.

Concretely, \tclocks maintains two kinds of variables: the $level$ of the local node with respect to any other node (as previously defined, which accounts for the corresponding direct view), and the largest $date$ at which a message carried to the local node through an indirect journey could have been emitted at that other node (accounting for indirect views). Note that even though the nodes clocks might not be synchronized, the fact that they all advance at the same rate and the latency is known to the nodes allows \tclocks to deliver information in the local time referential.

Technically, the abstraction materializes as an intermediate layer between the network and some higher algorithm (see Figure~\ref{fig:relation}), which is informed by means of generating the two following events: {\small \tt levelChanged(src, level, proxy)}, reflecting the evolution of a direct view, and {\small \tt dateImproved(src,date,proxy)}, reflecting that of an indirect view. (Both events are further described below.)

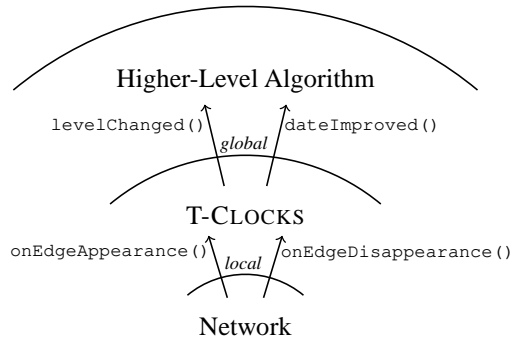
\begin{figure}[h]
  \centering
\begin{tikzpicture}[yscale=0.99]
  \tikzstyle{every node}=[font=\normalsize]
  \path (0,.5) node (dtn) {Network};
  \path (0,2) node (vc) {\tclocks};
  \path (0,3.8) node[sloped,bend right] (client) {Higher-Level Algorithm};
  \tikzstyle{every node}=[font=\scriptsize, inner sep=1pt]
  \tikzstyle{every path}=[semithick, shorten <=4pt, shorten >=1pt]
  \draw ([xshift=-.25cm]vc.north) edge[->] node[pos=.75,left] {{\tt levelChanged()}} node[pos=.55, right, inner sep=2pt]{\it global}([xshift=-.55cm]client.south);
  \draw ([xshift=.25cm]vc.north) edge[->] node[pos=.75,right] {{\tt dateImproved()}} ([xshift=.55cm]client.south);
  \draw ([xshift=-.2cm]dtn.north) edge[->] node[pos=.75,left] {{\tt onEdgeAppearance()}} node[pos=.6, right]{\it \ local} ([xshift=-.5cm]vc.south);
  \draw ([xshift=.2cm]dtn.north) edge[->] node[pos=.75,right] {{\tt onEdgeDisappearance()}} ([xshift=.5cm]vc.south);
  \tikzstyle{every path}=[semithick]
  \draw (50:4.8) arc (50:130:4.8);
  \draw (50:2.8) arc (50:130:2.8);
  \draw (50:1.2) arc (50:130:1.2);
\end{tikzpicture}
\caption{\label{fig:relation} \tclocks as an abstraction to track temporal views.}
\end{figure}

Given a direct or indirect view, we call {\em proxy} the local neighbor responsible for making the view evolve, i.e., the last node before the local node in the corresponding journey (this information is crucial to learn routing paths in the network).
Notifications do not only occur when a level or a date has been updated, but also occur when a proxy has changed (which may happen without any change of view, \textit{e.g.} when two local neighbors are providing the same direct view relative to a given remote node, and one of them disappears). \tclocks are used in the spirit of the {\em observer-observable} design pattern for Object-Oriented development~\cite{Wolfgang94} between two objects, namely the \tclocks algorithm (observable) and the higher algorithm (observer). The higher algorithm subscribes to the events of the \tclocks algorithm by calling a $register()$ function. Then it is notified of both types of events by means of calls to two functions $levelChanged()$ and $dateImproved()$ to be implemented on the higher algorithm and that encapsulate the desired response to the events. The notifications are raised as follows:

\begin{itemize}
\item $levelChanged($Node\ $src$, Integer $level$, Node $proxy)$: called whenever the proxy or direct view with respect to another node (also called {\em source}) has changed. This information is thus received as parameters of the call. 
\item $dateImproved($Node\ $src$, Integer $date$, Node $proxy)$, called whenever the indirect view relative to another node ($src$) has increased if and only if this date is larger than the direct view (that is, larger than $now()-level\times \zeta$ for the corresponding source).
\end{itemize}

\tclocks are typically assumed to run independently from the higher algorithm, and we will assume that it is already running when our higher algorithm registers to it. Note that using \tclocks does {\em hide} a substantial amount of complexity to the higher algorithm, and indeed the additive cost of our fastest broadcast tree construction algorithm is marginal. The complexity of \tclocks is unknown as of today, but the motivation is great to characterize and improve it, as any such improvement would have direct repercussions on the performance of all algorithms that use it as building block.

\section{Learning fastest broadcast trees in periodic TVGs}
\label{sec:fastest}

We now describe how to build fastest broadcast trees in periodically-varying graphs. This problem can be tackled in two steps. First, having the emitter learn {\em when} the broadcast has the potential to be the {\em fastest}, then build a {\em foremost} broadcast tree relative to any of the corresponding emission dates (by definition, nothing better can be done than a foremost broadcast tree once the best dates are known). Building a {\em foremost} broadcast tree relative to a given emission date does not require specific or difficult processes. This can be done by means of a flooding where all the nodes record which of their neighbor gave them the message first, followed by local acknowledgments of these relations. Thus, the rest of the section focuses on the mechanisms by which a node can learn {\em when} the broadcast has the potential to be the fastest ({\it i.e.,} its moment of {\em minimum temporal eccentricity}) using \tclocks.

The temporal eccentricity (or simply eccentricity below) of a node $u$ at date $t$ is formally defined as 
\begin{equation}
  \label{eq:eccentricity}
  ecc_{u}(t) = max\{\temp{d}_{u,t}(v):v \in V\},
\end{equation}
that is, the maximum among all temporal distances (or simply distances below) from $u$ to all other nodes at time $t$, see Section~\ref{sec:distances} for definitions. Informally speaking, it is the duration of the ``slowest'' foremost journey departing at $t$ from node $u$.

\subsection{High-level informal strategy}

The algorithm consists in inferring (and recording) temporal {\em distances} at every node relative to a given emitter, based on the evolution of temporal {\em views} monitored through \tclocks. Precisely, for a given emitter $u$, every node $v \ne u$ infers $\temp{d}_{u}(v)$ from $\phi_{v}(u)$ over one period and records it in a table called {\em distance table}. Since we deal with continuous-time and possibly overlapping contacts, this information is recorded as a set of {\em intervals} that correspond to the different phases of evolution of the distance (discrete or continuous). The distance tables of all nodes are then opportunistically  aggregated (arbitrarily) along a tree rooted in $u$. The aggregation of a children table consists of a segment-wise {\em maximum} against the local distance table (segments can be artificially split to match). Once the emitter has aggregated the table of its last child, the final result corresponds to its eccentricity over time (Equation~\ref{eq:eccentricity}). It finally selects any of the minimum values as preferred initiation date for the intended broadcast, then terminates. 

\subsection{Learning the temporal eccentricities over a period}

Let us first observe that learning the minimum temporal eccentricities in distributed networks is an interesting problem in its own right, and even though we use it here as a primitive for broadcasting, it could be used as well for other tasks, such as electing a leader based on its ability to reach all others quickly.

We now describe how tables of distances are learnt on the receptor side using \tclocks. Based on the relationship between temporal view and temporal distance, we establish three key properties through the following Lemmas. In what follows we use interchangeably the terms {\em emission date} and {\em initiation date}, both referring to the time when a potential message is sent at the emitter.

\begin{lemma}
  \label{lem:slope}
  Every discrete increase of $\phi_v(u)$ by value $k$ corresponds to a continuous decrease of $\temp{d}_u(v)$ of duration $k$.
\end{lemma}
\begin{proof}
  A discrete increase by $k$ of the view at time $t$ (that is, $\phi_{v,t}(u) = \phi_{v,t-\epsilon}(u) + k$), implies the existence of a journey $\J$ from $u$ whose departure is $\phi_{v,t}(u)$ and arrival is $t$. Switching to the emitter viewpoint, the considered increase implies that {\em no} journey starting during $[t_1=\phi_{v,t-\epsilon}(u), t_2=\phi_{v,t}(u))$ could have arrived before $t$ (otherwise such a journey would imply an intermediate increase of the view by less than $k$). Therefore, the temporal distance at any point in $[t_1, t_2)$ is fully determined by the arrival of $\J$, and thus for all $t'$ in $[t_1,t_2)$, we have $\temp{d}_{u,t'}(v) = \temp{d}_{u,t_2}(v) + (t_2-t')$, which corresponds to a continuous decrease of the distance during $t_2-t_1=k$ time units.\qed
\end{proof}

\begin{lemma}
  \label{lem:flat}
  Each continuous increase of $\phi_v(u)$ during $k$ time units corresponds to a stagnation of $\temp{d}_u(v)$ during $k$ time units.
\end{lemma}
\begin{proof}
  Continuous increases of the view are due to continuums of direct journeys of same level. Such increase during some interval $[t, t+k)$ thus implies a continuum of direct journeys departing over $[\phi_{v,t}(u), \phi_{v,t}(u)+k)$ (since $\zeta$ is constant). We thus have $\forall t_e \in [\phi_{v,t}(u),\phi_{v,t}(u)+k), \temp{d}_{u,t_e}(v)=level \times \zeta$ (where $level$ is the level of the considered journeys), which corresponds to a constant during $k$ time units.\qed
\end{proof}

\begin{lemma}
  \label{lem:complete}
  All initiation dates are covered either by a discrete or a continuous increase of the view.
\end{lemma}
\begin{proof}
  By nature of the view which is {\em past-inclusive}, {\it i.e.,} the fact that a message emitted at time $t_e$ could have arrived by some time $t_r$ implies that {\em any} message emitted before $t_e$ could have also arrived by $t_r$. (In essence, there is no ``gap''.) Besides, an increase is necessarily discrete or continuous.\qed
\end{proof}

Combination of Lemmas~\ref{lem:slope},~\ref{lem:flat} and~\ref{lem:complete} allows us to state the following general theorem on temporal distances (with constant edge latency $\zeta$).
\begin{theorem}
  \label{th:distance}
  The evolution of $\hat{d}_{u}(v)$ can be fully captured by a sequence of segments of two possible types: {\em flat} segments (stagnation of the value) and {\em slope} segments (continuous decrease of the value). This sequence can be inferred at $v$ by associating every continuous ({\it resp.} discrete) increase of $\phi_{v}(u)$ to a flat ({\it resp.} slope) segment of $\hat{d}_{u}(v)$.
\end{theorem}

This strategy is the one considered by Algorithm~\ref{algo:distance}, whose underlying principle consists in detecting (and recording) every transition between two segments of distance by means of transitions in the evolution of the temporal view (using \tclocks). Each transition is recorded as a triplet containing i) an emission date, ii) the corresponding value (distance), and iii) the type of segment {\em starting} at this emission date ({\em flat} or {\em slope}).
An example of such sequence is given on Figure~\ref{fig:distance-a-c}, representing the distance from node $a$ to node $c$ in the example TVG of Figure~\ref{fig:triangle}. 

%%%    TRIANGLE..
\begin{figure}[h]
  \centering
  \begin{tikzpicture}[scale=1.8]
    \clip (-1.2,-1.4) rectangle (1.2,.12);
    \tikzstyle{every node}=[draw, circle, inner sep=1.5pt]
    \path (0,0) node[very thick] (a){$a$};
    \path (a)+(-58:1.3) node (c){$c$};
    \path (a)+(-122:1.3) node (b){$b$};
    \tikzstyle{every node}=[font=\scriptsize]
    \draw (a)--node[midway,left]{$[0,30)$}(b);
    \draw (a)--node[midway,right]{$[20,60)$}(c);
    \draw (b)--node[midway,below]{$[10,40)\cup [70,80)$}(c);
  \end{tikzpicture}
  \caption{\label{fig:triangle}Example of periodic TVG ({\it with period $p=100$ and edge latency~$\zeta=1$}).}
\end{figure}
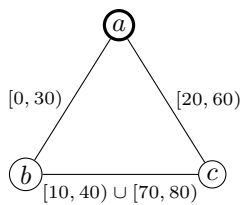

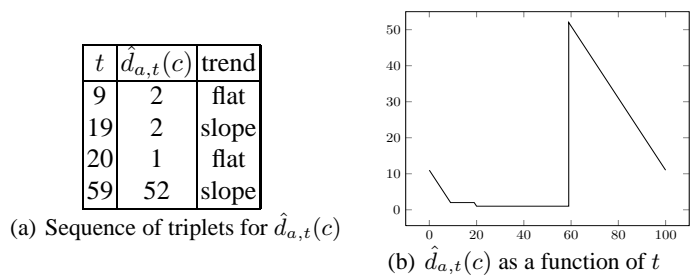
\begin{figure}[h]
  \centering
\subfigure[Sequence of triplets for $\temp{d}_{a,t}(c)$]{
  \label{fig:table-a-c}
  ~~~~~~~~~
  \begin{minipage}[c]{3cm}
  \begin{tabular}{|c|c|c|}
    \hline
    $t$&$\temp{d}_{a,t}(c)$&trend\\\hline
    9&2&flat\\
    19&2&slope\\
    20&1&flat\\
    59&52&slope\\\hline
  \end{tabular}

  \end{minipage}
  ~~~
}\hspace{10pt}
\subfigure[$\temp{d}_{a,t}(c)$ as a function of $t$]{
  \label{fig:curve-a-c}
  \begin{minipage}[c]{3.8cm}
    \begin{tikzpicture}[xscale=.55,yscale=.5]
      \begin{axis}[ymax=55]
        \addplot[mark=none] plot coordinates {
          (0,11)
          (9,2)
          (19,2)
          (20,1)
          (59,1)
          (59,52)
          (100,11)
        };
      \end{axis}
    \end{tikzpicture}
  \end{minipage}
}
\caption{\label{fig:distance-a-c}Temporal distance from $a$ to $c$, as a function of the emission date.}
\end{figure}

\begin{algorithm}[tbh]
  \small
  \begin{algorithmic}[1]
    \medskip
    \STATE $Map <$$Date, Distance$$>\ table \gets \emptyset$
    \STATE $Integer\ currentLevel \gets +\infty$
    \STATE $Date\ startD \gets nil$
    \STATE $Date\ pendingED \gets nil$
    \codetitle{\textbf{\textit{init()}}}
    \STATE $TClocks.register()$
 
    \codetitle{\textbf{\textit{dateImproved\,(Node\ $src$, Date $date$, Node $proxy$)}}}
    \IF {$src=emitter$}
    \STATE $update(date)$
    \ENDIF

    \codetitle{\textbf{\textit{levelChanged\,(Node $src$, Level $level$, Node $proxy$)}}}
    \IF {$src=emitter$}
    \STATE $update(now()-level \times \zeta)$
    \STATE $currentLevel \gets level$
    \ENDIF

    \codetitle{\textbf{\textit{update\,(Date $newED$)}}}
    \IF {$startD = nil$}
    \STATE $startD \gets now()$
    \STATE $pendingED \gets newED$
    \ELSE
    \STATE $updateFlat()$
    \STATE $updateSlope(newED)$
    \IF {$now() = startD + p$}\label{code:term1}
    \STATE $terminate$\label{code:term2}
    \ENDIF
    \ENDIF
    
    \codetitle{\textbf{\textit{updateFlat\,()}}}
    \IF {$currentLevel < +\infty$}\label{code:bestED1}
    \STATE $Date\ bestED \gets now() - currentLevel \times \zeta$
    \IF {$bestED > pendingED$}\label{code:bestED2}
    \STATE \commenttt{A flat segment is detected}
    \STATE $table.add(pendingED,currentLevel \times \zeta,"flat")$
    \STATE $pendingED \gets bestED$
    \ENDIF
    \ENDIF

    \codetitle{\textbf{\textit{updateSlope\,(Date $newED$)}}}
    \IF {${newED} > pendingED$}\label{code:newED}
    \STATE \commenttt{A slope segment is detected}
    \STATE $table.add(pendingED,\ now()-pendingED,\ "slope")$\label{code:slope}
    \STATE $pendingED \gets newED$
    \ENDIF
  \end{algorithmic}
  \caption{\label{algo:distance}Computing temporal distances at a receptor, relative to a given emitter.}
\end{algorithm}
The subtlety here is that the value of the distance at the beginning of a slope segment, as well as the duration of both types of segments, becomes known only {\em at the end} of the corresponding temporal view segment; therefore, the algorithm always records values relative to a previously pending emission date ($pendingED$). Precisely, whenever an event related to the temporal view occurs, whether it be caused by the arrival of a better indirect journey ($dateImproved()$) or a change in the level ($levelChanged()$), the same update function is called involving the following sub-routines (see Algorithm~\ref{algo:distance}):
\begin{list}{\labelitemi}{\leftmargin=.5em}
\item $updateFlat()$: If the current level ($currentLevel$) was not infinite, then the corresponding continuum may have delivered new emission dates (checked in lines~\ref{code:bestED1} to~\ref{code:bestED2}). If this is the case, then a flat segment is inserted for the pending date using distance value $currentLevel\times \zeta$ (see proof~\ref{lem:flat}).
\item $updateSlope()$: If the new event implies a discrete improvement of the view (line~\ref{code:newED}), then a slope segment must be created, and only then the distance at the pending emission date becomes known. The segment being a continuous decrease, the value corresponds to the distance at the newly received emission date plus the time elapsed between that date and the pending date (see proof~\ref{lem:slope}), {\it i.e.,} {\small $(now()-newED)+(newED-pendingED) = now()-pendingED$} (see line~\ref{code:slope}).
\end{list}

Observe that nothing prevents a same event from inducing both a flat {\em and} a slope segments, as is the case in the triangle {\small TVG} of Figure~\ref{fig:triangle} when $c$'s level relative to $a$ changes at time $21$. Table~\ref{tab:other-start} shows an example of execution trace corresponding to the most relevant steps at node $c$ (relative to emitter $a$) in the example of Figure~\ref{fig:distance-a-c}. The beginning of execution at node $c$ was arbitrarily set anytime between date $21$ and date $60$ (modulo~$p=100$).\bigskip

\begin{table}[h]
  \centering
  {\footnotesize
    \noindent
    \begin{tabular}{c|c|l}
      Date&Event&Action\\\hline
      60&$levelChanged(+\infty)$&$startD \gets 60$\\
      & &$pendingED \gets 59$\\
      & &$currentLevel \gets +\infty$\\
      111&$levelChanged(2)$&$updateFlat()$\\
      & &$updateSlope(109)$\\
      & &\boldmath$table.add(59,52,"slope")$\\
      & &$pendingED \gets 109$\\
      & &$currentLevel \gets 2$\\
      121&$levelChanged(1)$&$updateFlat()$\\
      & &\boldmath$table.add(109,2,"flat")$\\
      & &$pendingED \gets 119$\\
      & &$currentLevel \gets 1$\\
      & &$updateSlope(120)$\\
      & &\boldmath$table.add(119,2,"slope")$\\
      & &$pendingED \gets 120$\\
      & &$currentLevel \gets 1$\\
      160&$levelChanged(+\infty)$&$updateFlat()$\\
      & &\boldmath$table.add(120,1,"flat")$\\
      & &$terminate$\\
    \end{tabular}
  }\medskip
\caption{\label{tab:other-start}Relevant traces at node $c$ relative to emitter $a$ (assuming an arbitrary start of execution between dates $21$ and $60$).}
\end{table}

\begin{theorem}
  \label{th:correctdistance}
The execution of Algorithm~\ref{algo:distance} in a periodically-varying graph
$\G$ with known period $p$, relative to emitter $u$, results in every node $v$ capturing correctly the evolution of $\temp{d}_{u}(v)$ over one complete period, in $O(p)$.
\end{theorem}
\begin{proof}
  The argument is based on the correctness of \tclocks~\cite{CFMS11} combined with the strategy of Theorem~\ref{th:distance}, which Algorithm~\ref{algo:distance} implements. Theorem~\ref{th:distance} states that each segment of evolution of the temporal distance corresponds to a segment of evolution of the temporal view (the precise characterization of this correspondence being given by the proofs of Lemma~\ref{lem:slope} and~\ref{lem:flat}). The correctness follows from detecting all transitions in evolution of the temporal view (events $levelChanged()$ and $dateImproved()$), guaranteed by \tclocks, and transposing the corresponding segments into segments of the evolution of the temporal distance by means of procedures $updateFlat()$ and $updateSlope()$ in Algorithm~\ref{algo:distance} (both of which are called for every such transition, and each of which checks for the need to create the corresponding type of segment according to Lemmas~\ref{lem:slope} and~\ref{lem:flat}). One period exactly after the first record is inserted, a last record is inserted in the distance table, then the algorithm terminates (lines~\ref{code:term1} and~\ref{code:term2}). Because the dates are taken modulo $p$, this last record completes the distance table relative to one complete period.\qed
\end{proof}

\subsection{Aggregating distance tables back to the emitter}
Distance tables relative to a given emitter are aggregated along a tree rooted at that node. Such a tree may be arbitrary and built, for instance, using the same strategy as mentioned above for single foremost broadcast trees ({\it i.e.,} flooding a dedicated message from the emitter and detecting from which neighbor this message is first received at every node, followed by local acknowledgments of the corresponding relations). Once a node has computed its distance table and aggregated the tables of all its children (if any), it sends the resulting table to its parent. The aggregation of a children table (described below) consists of a segment-wise {\em maximum} among both tables for all emission dates (see Figure~\ref{fig:curves} for a visual illustration). 
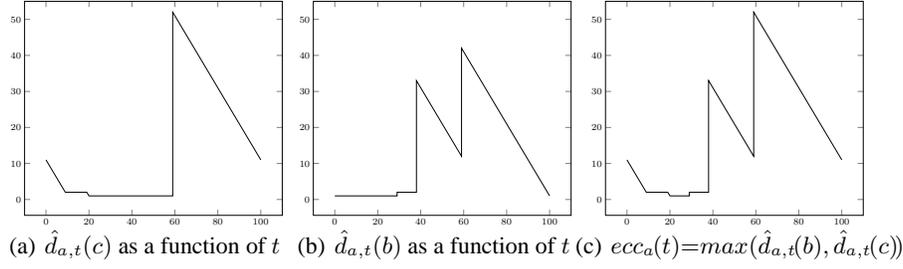
\begin{figure}[h]
  \centering
  \scriptsize
  \subfigure[$\temp{d}_{a,t}(c)$ as a function of $t$]{
    \label{fig:curve-a-c}
    \begin{tikzpicture}[xscale=.5,yscale=.5]
      \begin{axis}[ymax=55]
        \addplot[mark=none] plot coordinates {
          (0,11)
          (9,2)
          (19,2)
          (20,1)
          (59,1)
          (59,52)
          (100,11)
        };
      \end{axis}
    \end{tikzpicture}
  }
  \subfigure[$\temp{d}_{a,t}(b)$ as a function of $t$]{
    \begin{tikzpicture}[xscale=.5,yscale=.5]
      \begin{axis}[ymax=55]
        \addplot[mark=none] plot coordinates {
          (0,1)
          (29,1)
          (29,2)
          (38,2)
          (38,33)
          (59,12)
          (59,42)
          (100,1)
        };
      \end{axis}
    \end{tikzpicture}
  }
  \hspace{-4pt}
  \subfigure[$ecc_a$\hspace{-1pt}$(t)$$=$$max (\temp{d}_{a,t}$\hspace{-1pt}$(b), \temp{d}_{a,t}$\hspace{-1pt}$(c)$\hspace{-1.5pt}$)$]{
    ~
    \begin{tikzpicture}[xscale=.5,yscale=.5]
      \begin{axis}[ymax=55]
        \addplot[mark=none] plot coordinates {
          (0,11)
          (9,2)
          (10,2)
          (19,2)
          (20,1)
          (29,1)
          (29,2)
          (38,2)
          (38,33)
          (59,12)
          (59,52)
          (100,11)
        };
      \end{axis}
    \end{tikzpicture}
    ~~~~~~~
  }
\caption{\label{fig:curves}Aggregation of distances. The left curve (distance from $a$ to $c$) is combined to the middle curve (distance from $a$ to $b$) in order to yield the eccentricity of node $a$ over one period.}
\end{figure}
Once the emitter has aggregated all its children tables within its own table (initially consisting of a single flat segment of value~$0$), the final result corresponds to the maximum distance in every point in time among all nodes, and therefore (Equation~\ref{eq:eccentricity}) represents the evolution its eccentricity over one period. Once this table known, it finally selects any of the dates at which the eccentricity is minimum (relatively to which the foremost broadcast tree is to be built), then terminates. We now describe the aggregation process in more detail.

The purpose of aggregation is to select the maximum value among all distances in every point in time. Aggregating two tables based on their segments would be relatively straightforward if the segments were horizontally aligned ({\it i.e.,} if the period was split into the same sequence of intervals in both tables) and vertically non-crossing ({\it i.e.,} given any pair of aligned segments, one of them remains higher than or equal to the other during the corresponding interval). In such an ideal case, the maximum operation between two aligned segments, say ${\cal S}_i=(t_i,\temp{d}_{t_i},trend_i)$ and ${\cal S}_i'=(t_i,\temp{d}'_{t_i},trend'_i)$ is straightforward; it may consist of selecting ${\cal S}_i$ iff
$$(\temp{d}_{t_i} > \temp{d}_{t_i'}) \vee (\temp{d}_{t_i} = \temp{d}_{t_i'} \wedge trend_i=flat)\mbox{; and ${\cal S}_i'$ otherwise. }$$

\begin{theorem}
  \label{th:correctmaximum}
  Given two segments whose intervals are aligned and values are non-crossing, the above logic selects the correct maximum segment.
\end{theorem}
\begin{proof}
  Let us examine separately the cases that $\temp{d}_{t_i} > \temp{d}_{t_i'}$, $\temp{d}_{t_i} < \temp{d}_{t_i'}$, or $\temp{d}_{t_i} = \temp{d}_{t_i'}$. If $\temp{d}_{t_i} > \temp{d}_{t_i'}$ ${\cal S}_i$ is selected irrespective of the trend (left clause). This condition is sufficient because segments cannot cross. If $\temp{d}_{t_i} < \temp{d}_{t_i'}$ the formula evaluates to false, leading to select ${\cal S}_i'$, which is correct for the same reason. If $\temp{d}_{t_i} = \temp{d}_{t_i'}$ then three cases are again possible: either both segments are flat (and therefore equals, ${\cal S}_i$ is arbitrarily selected), or both segments are a slope (and therefore equals, ${\cal S}_i'$ is arbitrarily selected), or ${\cal S}_i$ is flat and ${\cal S}_i'$ is a slope (a flat segment starting at the same high as a decreasing slope will necessarily remain higher; ${\cal S}_i$ is therefore selected), or ${\cal S}_i$ is a slope and ${\cal S}_i'$ is flat (symmetrical case, ${\cal S}_i'$ will be selected).\qed
\end{proof}

In reality, the segments do cover intervals of various sizes and temporal location, and in some cases the corresponding distance values may ``cross'' even when the intervals covered by two segments are aligned ({\it e.g.} a slope segment starting at a slightly higher value than a flat segment). Rather than complicating the above logic of aggregation, we pre-process the two tables before aggregation as follows:

\begin{list}{\labelitemi}{}
\item[1)] {\it Split the segments so as to align both tables:} Splitting a segment $(t_i,\temp{d}_{t_i},trend)$ at date $t_i'$ comes to insert a subsequent entry $(t_i',\temp{d}_{t_i'},trend)$ such that $\temp{d}_{t_i'}=\temp{d}_{t_i}$ if $trend=flat$ and $\temp{d}_{t_i'}=\temp{d}_{t_i}-(t_i'-t_i)$ if $trend=slope$. Each table undergoes such a split relative to every index date ($t_i's$) that exists only in the other table. An additional split relatively to time $0$ may also be added for convenience (if not already present).
\item[2)] {\it Split further to eliminate crossing values:} Given two aligned segments of size $l$, if one is flat and the other a slope, say $(t,\temp{d}_{f_t},flat)$ and $(t,\temp{d}_{s_t},slope)$, and $0 < \temp{d}_{s_t} - \temp{d}_{f_t} < l$, then both segments are split at $t+(\temp{d}_{s_t} - \temp{d}_{f_t})$ precisely. 
\end{list}

\begin{theorem}
  \label{th:correctpreprocessing}
  These pre-processing steps produce tables whose segments are aligned and non-crossing.
\end{theorem}
\begin{proof}
  First of all, the fact that a split preserves the consistency of a table is clear from the formulas in step 1 ({\it i.e.,} using the same distance value if the segment is flat; decreasing it otherwise by an amount equal to the time elapsed since the beginning of the segment). Since each table undergoes a split with respect to every index date that exists only in the other table, the resulting tables must contain the same index dates (aligned segments). As for the crossing segments, let us first observe that two flat segments cannot cross since their value is a constant, neither can two slope segments because their value decrease at a same rate (the rate of time). Therefore, crossings may only occur between segments of different trends. Let $\delta=\temp{d}_{s_t} - \temp{d}_{f_t}$ be the difference of initial value between the slope segment and the flat segment. The segments cannot cross if the slope segment is already below the flat segment at the beginning of the interval (a slope is always decreasing), neither can they if $\delta$ is larger than $l$ because the final value of the slope segment will still be above that of the flat segment. Therefore, crossings only occur if $0 < \delta < l$. Besides, both types of segment trend being linear, two given segments cannot cross more than once. It is therefore sufficient to split them with respect to their crossing point $t+\delta$.\qed
\end{proof}
Figure~\ref{fig:tables} illustrates the aggregation of two distance tables that relate again to the same example, {\it i.e.,} $\temp{d}_{a,t}(b)$ and $\temp{d}_{a,t}(c)$ in the triangle {\small TVG} of Figure~\ref{fig:triangle}. Note that no situation of crossing segments occurred to be handled.
\begin{figure}[h]
  \centering
  \begin{tikzpicture}
    \path (0,0) node[below, font=\small, text width=5cm] (b){
      \begin{tabular}{|c|c|c|}
        \hline
        ED&Dist&Type\\\hline
        0&1&flat\\
        29&2&flat\\
        38&33&slope\\
        59&42&slope\\\hline
      \end{tabular}
    };
    \path (b.north) node[below=15pt]{~~~~$\oplus$~~};
    \path (b.north east) node[right=20pt, below, font=\small, text width=5cm] (c){
      \begin{tabular}{|c|c|@{~}c@{~}|}
        \hline
        ED&Dist&Type\\\hline
        \it (0)&\it (11)&\it (slope)\\
        9&2&flat\\
        19&2&slope\\
        20&1&flat\\
        59&52&slope\\\hline
      \end{tabular}
    };
    \path (c.north) node[below=20pt,right=5pt]{$=$};
    \path (c.north east) node[right=0pt, below, font=\small, text width=3cm] (a){
      \begin{tabular}{|c|c|c|}
        \hline
        ED&Ecc&Type\\\hline
        0&11&slope\\
        9&2&flat\\
        19&2&slope\\
        20&1&flat\\
        29&2&flat\\
        38&33&slope\\
        59&52&slope\\\hline
      \end{tabular}
    };
  \end{tikzpicture}
\caption{\label{fig:tables}Aggregation of two distance tables (corresp. to $\temp{d}_{a,t}(b)$ and $\temp{d}_{a,t}(c)$ in Fig.~\ref{fig:triangle}).}
\end{figure}
Based on the final aggregation, the emitter can decide when to initiate the intended fastest broadcast --~here, anytime between dates $20$ (inclusive) and date $29$ (exclusive) modulo $p$.

\section{Concluding Remarks and Open Problems}
The problem of computing fastest journeys in delay-tolerant networks was defined and solved in~\cite{BFJ03} in its combinatorial variant ({\it i.e.,} by a centralized algorithm that has full knowledge of the network schedule). We formulated this problem in a distributed setting, and solved it in the general case that contacts have arbitrary durations and can possibly overlap with each other. This case is more complex than the case with punctual contacts, because both direct and indirect journeys can co-exist in the network. Using the \tclocks abstraction from~\cite{CFMS11}, we showed how nodes could learn their own temporal eccentricity in such a context in the particular case of periodically-varying graph. Based on this information, the task of building fastest broadcast trees reduces to that of building a {\em foremost} broadcast tree relative to (any of) the date of minimum temporal eccentricity.

The purpose of this paper was primarily to demonstrate the feasibility of this problem. As far as complexity is concerned, most of the communication cost is encapsulated at the level of \tclocks. Characterizing and improving this cost thus requires to characterize and improve that of \tclocks. However the complexity of \tclocks is still unknown. This open problem is regarded as a relevant research avenue, all the more that \tclocks is used as a building block to solve various concrete problems. 
Contrary to most distributed algorithms for static networks, the complexity of an algorithm in the context of interest does not only depend on the number of nodes and edges, but is strongly dependent on the number of {\em topological events} during the execution (in fact, a vast majority of communications and computations are precisely triggered by these events). It seems therefore reasonable to start by looking at the way various network schedules could impact complexity.

% \bibliographystyle{plain}
% \bibliography{fastest}

\begin{thebibliography}{10}

\bibitem{AE84}
B.~Awerbuch and S.~Even.
\newblock Efficient and reliable broadcast is achievable in an eventually
  connected network.
\newblock In {\em Proceedings of the 3rd ACM symposium on Principles of
  distributed computing}, pages 278--281, Vancouver, Canada, 1984. ACM.

\bibitem{Berman96}
K.A. Berman.
\newblock {Vulnerability of scheduled networks and a generalization of Menger's
  Theorem}.
\newblock {\em Networks}, 28(3):125--134, 1996.

\bibitem{BF03}
S.~Bhadra and A.~Ferreira.
\newblock Complexity of connected components in evolving graphs and the
  computation of multicast trees in dynamic networks.
\newblock In {\em Proceedings of the 2{nd} International Conference on Ad Hoc,
  Mobile and Wireless Networks (AdHoc-Now)}, pages 259--270, Montreal, Canada,
  2003. Springer.

\bibitem{BFJ03}
B.~{Bui-Xuan}, A.~Ferreira, and A.~Jarry.
\newblock Computing shortest, fastest, and foremost journeys in dynamic
  networks.
\newblock {\em International Journal of Foundations of Computer Science},
  14(2):267--285, April 2003.

\bibitem{CCF09}
A.~Casteigts, S.~Chaumette, and A.~Ferreira.
\newblock Characterizing topological assumptions of distributed algorithms in
  dynamic networks.
\newblock In {\em Proceedings of the 16th International Colloquium on
  Structural Information and Communication Complexity (SIROCCO)}, pages
  126--140, Piran, Slovenia, 2009. Springer.
\newblock (Full version on {\it arXiv:1102.5529}).

\bibitem{CFMS10}
A.~Casteigts, P.~Flocchini, B.~Mans, and N.~Santoro.
\newblock Deterministic computations in time-varying graphs: Broadcasting under
  unstructured mobility.
\newblock In {\em Proceedings 5th IFIP Conference on Theoretical Computer
  Science (TCS)}, pages 111--124, Brisbane, Australia, 2010. Springer.

\bibitem{CFMS11}
A.~Casteigts, P.~Flocchini, B.~Mans, and N.~Santoro.
\newblock Measuring temporal lags in delay-tolerant networks.
\newblock In {\em 25th IEEE International Parallel \& Distributed Processing
  Symposium (IPDPS)}, pages 209--218, Anchorage, USA, May 2011. IEEE.

\bibitem{CFQS11}
A.~Casteigts, P.~Flocchini, W.~Quattrociocchi, and N.~Santoro.
\newblock Time-varying graphs and dynamic networks.
\newblock In {\em Proceedings of the 10th International Conference on Ad Hoc
  Networks and Wireless (ADHOC-NOW)}, pages 346--359, Paderborn, Germany, July
  2011.

\bibitem{ChMMD08}
A.~Chaintreau, A.~Mtibaa, L.~Massoulie, and C.~Diot.
\newblock The diameter of opportunistic mobile networks.
\newblock {\em Communications Surveys \& Tutorials}, 10(3):74--88, 2008.

\bibitem{DubGV03}
H.~Dubois-Ferriere, M.~Grossglauser, and M.~Vetterli.
\newblock {Age matters: efficient route discovery in mobile ad hoc networks
  using encounter ages}.
\newblock In {\em Proceedings ACM International Symposium on Mobile Ad Hoc
  Networking \& Computing (MobiHoc)}, page 266, 2003.

\bibitem{FKMS11}
P.~Flocchini, M.~Kellett, P.C. Mason, and N.~Santoro.
\newblock Searching for black holes in subways.
\newblock {\em Theory of Computing Systems}, 50(1):158--184, 2012.

\bibitem{FMS09}
P.~Flocchini, B.~Mans, and N.~Santoro.
\newblock Exploration of periodically varying graphs.
\newblock In {\em Proceedings 20th International Symposium on Algorithms and
  Computation (ISAAC)}, pages 534--543, 2009.

\bibitem{GroV03}
M.~Grossglauser and M.~Vetterli.
\newblock {Locating nodes with EASE: Last encounter routing in ad hoc networks
  through mobility diffusion}.
\newblock In {\em Proceedings 22nd Conference on Computer Communications
  (INFOCOM)}, volume~3, pages 1954--1964, San Francisco, USA, 2003. IEEE.

\bibitem{IW11}
D.~Ilcinkas and A.~Wade.
\newblock On the power of waiting when exploring public transportation systems.
\newblock {\em Proceedings of the 15th International Conference on Principles
  of Distributed Systems (OPODIS)}, pages 451--464, 2011.

\bibitem{JMR10}
P.~Jacquet, B.~Mans, and G.~Rodolakis.
\newblock {Information propagation speed in mobile and delay tolerant
  networks}.
\newblock {\em IEEE Transactions on Information Theory}, 56(1):5001--5015,
  2010.

\bibitem{JFP04}
S.~Jain, K.~Fall, and R.~Patra.
\newblock {Routing in a delay tolerant network}.
\newblock In {\em Proceedings of Conference on Applications, Technologies,
  Architectures, and Protocols for Computer Communications (SIGCOMM)}, pages
  145--158, 2004.

\bibitem{JLW07}
E.P.C. Jones, L.~Li, J.K. Schmidtke, and P.A.S. Ward.
\newblock {Practical routing in delay-tolerant networks}.
\newblock {\em IEEE Transactions on Mobile Computing}, 6(8):943--959, 2007.

\bibitem{KKK00}
D.~Kempe, J.~Kleinberg, and A.~Kumar.
\newblock {Connectivity and inference problems for temporal networks}.
\newblock In {\em Proceedings 32nd ACM Symposium on Theory of Computing}, pages
  504--513, Portland, USA, 2000. ACM.

\bibitem{KerO09}
A.~Ker{\"a}nen and J.~Ott.
\newblock {DTN over aerial carriers}.
\newblock In {\em Proceedings 4th ACM Workshop on Challenged Networks}, pages
  67--76, Beijing, China, 2009. ACM.

\bibitem{KMR05}
A.~Khelil, P.J. Marron, and K.~Rothermel.
\newblock {Contact-based mobility metrics for delay-tolerant ad hoc
  networking}.
\newblock In {\em Proceedings 13th IEEE Int. Symp. on Modeling, Analysis, and
  Simulation of Computer and Telecommunication Systems}, pages 435--444, 2005.

\bibitem{KosKW08}
G.~Kossinets, J.~Kleinberg, and D.~Watts.
\newblock {The structure of information pathways in a social communication
  network}.
\newblock In {\em Proceedings 14th International Conference on Knowledge
  Discovery and Data Mining (KDD)}, pages 435--443, Las Vegas, USA, 2008. ACM.

\bibitem{KLO10}
F.~Kuhn, N.~Lynch, and R.~Oshman.
\newblock {Distributed computation in dynamic networks}.
\newblock In {\em Proceedings of the 42nd ACM Symposium on Theory of Computing
  (STOC)}, pages 513--522, Cambridge, USA, 2010. ACM.

\bibitem{LDS03}
A.~Lindgren, A.~Doria, and O.~Schel\'{e}n.
\newblock Probabilistic routing in intermittently connected networks.
\newblock {\em SIGMOBILE Mob. Comput. Commun. Rev.}, 7(3):19--20, 2003.

\bibitem{LW09b}
C.~Liu and J.~Wu.
\newblock Scalable routing in cyclic mobile networks.
\newblock {\em IEEE Transactions on Parallel and Distributed Systems},
  20(9):1325--1338, 2009.

\bibitem{MM11}
B.~Mans and L.~Mathieson.
\newblock On the treewidth of dynamic graphs.
\newblock {\em CoRR}, abs/1112.2795, 2011.

\bibitem{Wolfgang94}
P.~Wolfgang.
\newblock {\em Design patterns for object-oriented software development}.
\newblock Reading, Mass.: Addison-Wesley, 1994.

\bibitem{YamK96}
M.~Yamashita and T.~Kameda.
\newblock Computing on anonymous networks: Part {I} and {II}.
\newblock {\em IEEE Trans. on Par. and Distributed Systems}, 7(1):69 -- 96,
  1996.

\end{thebibliography}

\end{document}